\documentclass[11pt,letterpaper]{article}
\usepackage[margin=1in]{geometry}  

\usepackage{epsfig}
\usepackage{amsmath,amsfonts,amssymb,color}
\usepackage{epsfig}
\usepackage{psfrag}
\usepackage{algorithm}
\usepackage{algpseudocode}
\usepackage{array}
\usepackage{epstopdf}
\usepackage{tikz, subcaption, cite}
\usepackage{float}
\usepackage{mathtools}
\usepackage{caption}
\usepackage{amsthm}
\usepackage{empheq}

\newtheorem{theorem}{Theorem}

\newtheorem{corollary}{Corollary}
\newtheorem{remark}{Remark}

\newtheorem{example}{Example}

\allowdisplaybreaks

\makeatletter
\newenvironment{breakablealgorithm}
  {
   \begin{center}
     \refstepcounter{algorithm}
     \hrule height.8pt depth0pt \kern2pt
     \renewcommand{\caption}[2][\relax]{
       {\raggedright\textbf{\ALG@name~\thealgorithm} ##2\par}%
       \ifx\relax##1\relax 
         \addcontentsline{loa}{algorithm}{\protect\numberline{\thealgorithm}##2}%
       \else 
         \addcontentsline{loa}{algorithm}{\protect\numberline{\thealgorithm}##1}%
       \fi
       \kern2pt\hrule\kern2pt
     }
  }{
     \kern2pt\hrule\relax
   \end{center}
  }
\makeatother
\algnewcommand{\LeftComment}[1]{\Statex \(\triangleright\) #1}

\title{\LARGE \bf Min-Max Tours for Task Allocation to Heterogeneous Agents}
\author{Amritha~Prasad$^*$,~Han-Lim~Choi$^\dagger$~and~Shreyas~Sundaram$^*$
\thanks{$^*$Amritha Prasad and Shreyas Sundaram are with the School of Electrical and Computer Engineering at Purdue University. Email: {\tt \{prasad20,sundara2\}@purdue.edu}\newline
$^\dagger$Han-Lim Choi is with the Department of Aerospace Engineering at Korea Advanced Institute of Science and Technology. Email: {\tt hanlimc@kaist.ac.kr}}%
}
\date{}

\begin{document}
\maketitle
\thispagestyle{empty}
\pagestyle{empty}

\begin{abstract}
We consider a scenario consisting of a set of heterogeneous mobile agents located at a depot, and a set of tasks dispersed over a geographic area. The agents are partitioned into different types. The tasks are partitioned into specialized tasks that can only be done by agents of a certain type, and generic tasks that can be done by any agent. The distances between each pair of tasks are specified, and satisfy the triangle inequality. Given this scenario, we address the problem of allocating these tasks among the available agents (subject to type compatibility constraints) while minimizing the maximum cost to tour the allocation by any agent and return to the depot. This problem is NP-hard, and we give a three phase algorithm to solve this problem that provides 5-factor approximation, regardless of the total number of agents and the number of agents of each type. We also show that in the special case where there is only one agent of each type, the algorithm has an approximation factor of 4. 
\end{abstract}

\section{Introduction}
Multi-robot systems will play a large role in a variety of modern and future applications including exploration, surveillance, search and rescue operations, cooperative control and operations in hazardous environments. In order to effectively utilize these multi-robot systems in such applications, it is necessary to allocate an appropriate set of tasks to each robot or agent in the system.  Such problems have been widely considered in the literature \cite{sundar2015exact, choi2009consensus, bertuccelli2009real, sadeghi2017heterogeneous, bullo2011dynamic}, most typically for the case where all agents are the same.  However, future multi-robot systems are also projected to have a large amount of diversity in terms of the capabilities of the agents, and the applications will consist of tasks that can only be done by agents that possess certain capabilities \cite{prorok2016adaptive, jones2006dynamically, labella2006division, kilgore2007mission, navalres2008auto}. We address this problem in this paper, namely allocating tasks efficiently to heterogeneous agents while meeting task-agent compatibility constraints.

Gerkey and Mataric \cite{gerkey2004formal} gives a taxonomy of task allocation in multi-robot systems. Prorok et. al. \cite{prorok2016adaptive} distributes a swarm of heterogeneous robots (of different types, with each type having different traits) among a set of tasks that require specialized capabilities in order to be completed. They optimize the transition rates for each type of robot so that the desired trait distribution is reached, but do not consider travel time between tasks. In this paper, we focus on task allocation to heterogeneous robots (of different types, with different capabilities) such that the task-agent compatibility constraints are satisfied and the maximum tour cost incurred by any agent is minimized. The existing literature has several works on task allocation to agents where the total tour cost (i.e., sum of tour costs incurred by all agents) is minimized. Bektas \cite{Bektas2006209} gives an overview on the work on the Traveling Salesperson Problem (TSP) with multiple (homogeneous) salespersons, where the sum of the cost of tours by all salespersons is minimized. Frieze \cite{frieze1983extension} extends Christofides' $\frac{3}{2}$-approximation algorithm \cite{christofides1976worst} for the single Traveling Salesperson Problem to the $k$-person Traveling Salesperson Problem (where all salespersons are identical and start at the same node). Xu and Rodrigues \cite{xu20153} provide a $\frac{3}{2}$-approximation algorithm for multiple TSP with a fixed number of depots. Malik, Rathinam and Darbha \cite{malik2007approximation} provide a $2$-approximation algorithm for a generalized multiple depot, multiple TSP with symmetric travel costs. Yadlapalli and Rathinam \cite{yadlapalli20123} provide a $3$-approximation algorithm for a case where two heterogeneous vehicles with associated travel costs start from distinct initial locations and are jointly required to visit a set of targets. Bae and Rathinam \cite{bae2016primal} provide a $2$-approximation algorithm for a special case of the above problem where the travel cost for one vehicle is at most that of the other.

As opposed to the above works that minimize the sum of the costs of each agent, the following works focus on minimizing the maximum tour cost incurred by any agent in a group of homogeneous agents. Frederickson, Hecht and Kim \cite{frederickson1976approximation} give tour splitting heuristics for the $k$-person variants of the Traveling Salesperson Problem. Given a tour constructed (using an $F$-approximation algorithm) on a set of nodes, they provide a $1 + F - \frac{1}{k}$ factor algorithm that outputs $k$ subtours such that the cost of the largest subtour is minimized. Yu and Liu \cite{yu2015improved} consider networks with multiple depots (or start nodes) and provide a $6$-approximation factor algorithm. Even et. al. \cite{even2004min} and Arkin et. al. \cite{arkin2006approximations} provide a $4$-approximation algorithm for a problem known as the min-max tree cover problem. This was improved by Khani and Salavatipour \cite{khani2014improved} who proposed a $3$-approximation algorithm. The recent work by Sadeghi and Smith \cite{sadeghi2017heterogeneous} give a decentralized auction based multi-robot task allocation where the objective is to minimize the total time taken to perform all tasks. They consider a case of heterogeneous robots, where the heterogeneity is captured by constraints on the motion of the robots. 

The above works focus either on minimizing the total cost incurred by all agents or on minimizing the maximum cost incurred by any agent in a set of agents, where all agents have the same functionality. In this paper, we combine the idea of heterogeneity in agent functionality with that of minimizing the maximum cost incurred by any agent. Specifically, consider a scenario where a set of tasks at different locations need to be executed; however, not all tasks can be done by all agents and certain task-agent compatibility constraints must be satisfied. Agents are partitioned into different types based on functionality of the agents. Tasks are partitioned into sets of type-specific tasks and generic tasks, where type-specific tasks can only be performed by agents of a given type and generic tasks can be performed by any agent.
 
To capture this scenario, we present the Heterogeneous Task Allocation Problem (HTAP) which aims to allocate a set of tasks among heterogeneous agents such that the maximum time to complete tasks by any agent is minimized. This is an important metric, especially when the tasks to be performed are time critical or when the quality of service is characterized by maximum delays. We find tours for each agents that start and end at a start node (e.g., home base), such that the task-agent compatibility constraints are met, and approximately minimize the maximum cost incurred by any agent to complete its tour.

The rest of this paper is organized as follows. In the next section, we state our problem formulation. We then present a naive algorithm and show that the approximation factor of this algorithm, although bounded, increases linearly with the number of agents. This motivates the need to develop better algorithms that perform well as the number of agents increase. We present two such algorithms in the subsequent section. We first develop an algorithm called \textsc{CycleSplit}. We leverage insights learned from \textsc{CycleSplit} and use it as a baseline to develop an alternate algorithm which we call \textsc{HeteroMinMaxSplit}. We show that these are 5-approximation algorithms for HTAP. We also show that in the special case where each of the heterogeneous agents are distinct, the proposed algorithms haves an approximation factor of 4.

\section{Problem Statement}
Consider a set of tasks $T$ that are required to be completed by a set of $k$ heterogeneous agents $A = \{A_1$, $A_2$,..., $A_k\}$. Each agent is one of $m$ {\it types}. Let $f:\{1,\ldots,k\}\to \{1,\ldots,m\}$ be a function that takes an agent number as input and outputs the type of that agent. For each $i \in \{1,2,\ldots,m\}$, let $m_i$ be the number of agents of type $i$, with $\sum_{i=1}^{m} m_i = k$. Let $T$ be composed of two broad classes of tasks: type-specific tasks and generic tasks. Type-specific tasks can be performed only by a specific type of agent, whereas generic tasks can be performed by any agent. Let $T_0$ denote the set of generic tasks and $T_i,~1\leq i \leq m$, denote the set of type-specific tasks that can be performed by agents of type $i$. Thus, $T = T_0~\cup~(\cup_{i=1}^{m}T_i)$. Let all agents start at a node $v_s$, called the start node (or home base). 

Consider a complete graph $G(V,E)$ with vertex set $V = T \cup \{v_s\}$ and edge set $E = \{(u,v):u,v \in V,~u \neq v\}$. Let each edge $e = (u,v) \in E$ have a weight $d(u,v)$ given by the distance between the nodes $u$ and $v$. Let the direct travel cost between two nodes be the weight of the edge connecting the two nodes in $G$. We assume the distances satisfy the triangle inequality. The cost of executing a task is assumed to be very small compared to travel costs and is hence neglected. A tour on a subset of nodes $V'\subseteq V$ is a closed path from $v_s$ through all nodes in $V'$ (in some order) ending at $v_s$. All tours start and end at the start vertex $v_s$. The cost of the tour is defined as the sum of weights of all edges on that tour. Let $C^*(V')$ denote the tour cost of the optimal tour on set $V'$.

The objective of the allocation problem is to partition the set of tasks $T$ among the agents subject to the task-agent compatibility constraints, such that the maximum cost among all agents to tour their allocated tasks is minimized. This is framed as the Heterogeneous Task Allocation Problem given below.

\vspace{10pt}
\textbf{Heterogeneous Task Allocation Problem (HTAP):}
\begin{equation*}
\begin{aligned}
&\min\limits_{S_1,S_2,\ldots,S_k \subseteq T}~\max\limits_{1 \leq j \leq k}~C^*(S_j) \\
&\text{subject to} \quad \cup_{j=1}^{k} S_j = T,\\
&\hspace{52pt} S_j = V_j \cup R_j,~\forall j \in \{1,2,\ldots,k\},\\
&\hspace{52pt} V_j \subseteq T_{f(j)},~R_j \subseteq T_0.
\end{aligned}
\label{HTAP}
\end{equation*}

In the above formulation, $S_j$ is the task set allocated to agent $A_j$, for $1 \leq j \leq k$. The first constraint implies that each task must be executed by some agent. The remaining conditions state that the task set allocated to each agent $A_j$ is a union of type-specific tasks $V_j$ (which is a subset of $T_{f(j)}$, where $f(j)$ is the type of agent $A_j$) and generic tasks $R_j$ (which is a subset of generic tasks $T_0$).

\begin{remark}
	Although we do not explicitly specify the condition that a task cannot be performed by more than one agent, minimization of the objective function implicitly ensures that the optimal solution does not have any node appearing in multiple subsets due to the triangle inequality.
	\label{remark1}
\end{remark}

Any instance of the Traveling Salesperson Problem (TSP) can be trivially reduced to an instance of the Heterogeneous Task Allocation Problem (HTAP) by setting the number of agents $k=1$ and all tasks to be generic. Thus, the HTAP is trivially NP-Hard, and has no polynomial time solution unless P = NP. Hence, in the rest of this paper we develop approximation algorithms for HTAP. An $\alpha$-approximation algorithm for a problem is an algorithm that efficiently finds a solution within a factor $\alpha$ of the optimum solution for every instance of that problem.

\section{A Naive Approximation Algorithm for the Heterogeneous Task Allocation Problem}
We start with the following simple algorithm to solve HTAP. In this naive algorithm, we allocate all tasks in $T$ to a single agent $A_1$. To every other agent, we allocate the type-specific tasks associated with that agent. Agent $A_1$ computes a tour using some approximation algorithm (e.g., Christofides' algorithm) on set $T$ (which includes type-specific tasks of all agents and all the generic tasks). Each other agent $A_j, ~j \neq 1$, computes a tour on its type-specific tasks using some approximation algorithm (e.g., Christofides' algorithm). Algorithm \ref{alg:NaiveAlg} describes this naive allocation.

\begin{breakablealgorithm}
	\caption{\textsc{NaiveAllocation} Algorithm}
	\label{alg:NaiveAlg}
	\begin{algorithmic}[1]
		\Procedure{NaiveAllocation}{$A,T,G$}
		\State \parbox[t]{\dimexpr 433pt}{Allocate type-specific tasks to each agent of that type.\strut}
		\State \parbox[t]{\dimexpr 433pt}{Allocate all generic tasks to agent $A_1$.\strut}
		\State \parbox[t]{\dimexpr 433pt}{Compute tour (starting and ending at node $v_s$) for the set of tasks allocated to each agent using an approximation algorithm.\strut}
		\State \parbox[t]{\dimexpr 433pt}{Return tours for each agent.\strut}
		\EndProcedure
	\end{algorithmic}
\end{breakablealgorithm}

We provide the approximation factor of this algorithm in the following theorem.

\begin{theorem}
	Suppose the algorithm that is used to compute the tour for each agent in line 4 of \textsc{NaiveAllocation} algorithm has approximation factor $\alpha$. Then, \textsc{NaiveAllocation} is an $\alpha k$ approximation algorithm for the Heterogeneous Task Allocation Problem (HTAP).
	\label{thm_naive}
\end{theorem}

\begin{proof}
	For $1\leq j \leq k$, let $S^*_j$ denote the set of tasks allocated to agent $A_j$ under the optimal allocation for HTAP. Note that $S^*_j \subseteq T~\forall j$, $S^*_j \cap S^*_i = \emptyset$ if $j \neq i$ and $\cup_{j=1}^{k} S^*_j = T$. The set $S^*_j$ is given by $S^*_j \subseteq T_{f(j)} \cup R^*_j$ where $R^*_j$ is the subset of tasks in $T_0$ allocated to agent $A_j$ under the optimal allocation policy. Note that $R^*_j \subseteq T_0~\forall j$, $R^*_j \cap R^*_i = \emptyset$ for $j \neq i$ and $\cup_{j=1}^{k} R^*_j = T_0$.
	
	For $1\leq j \leq k$, let $S_j$ denote the set of tasks allocated to agent $A_j$ by \textsc{NaiveAllocation}.  Note that $S_j \subseteq T~\forall j$, $S_1 = T$ and $S_j =  T_{f(j)},~j\neq1$.	Let $C^*(\cdot)$ denote the cost to optimally tour a given set of tasks and let $C_{NA}(\cdot)$ denote the tour costs returned by  the \textsc{NaiveAllocation} algorithm.
	
	The approximation ratio $R$ for the \textsc{NaiveAllocation} algorithm is given by
	\begin{equation*}
		\begin{split}
		R  =~& \frac{\max\limits_{1 \leq j \leq k} C_{NA}(S_j)}{\max\limits_{1 \leq j \leq k} C^*(S^*_j)} = \frac{C_{NA}(S_1)}{\max\limits_{1 \leq j \leq k} C^*(S^*_j)} = \frac{C_{NA}(T)}{\max\limits_{1 \leq j \leq k} C^*(S^*_j)}\\
		& \leq \frac{\alpha C^*(T)}{\max\limits_{1 \leq j \leq k} C^*(S^*_j)}
	      \leq \frac{\alpha ( \sum_{j=1}^{k}C^*(S^*_j))}{\max\limits_{1 \leq j \leq k} C^*(S^*_j)} \leq \alpha k.
		\end{split}
	\end{equation*}
	The inequality $C_{NA}(T) \leq \alpha C^*(T)$ follows from the fact that an $\alpha$-approximation algorithm is used to compute the tour in line 4 of \textsc{NaiveAllocation}. Also, $C^*(T) \leq \sum_{j=1}^{k} C^*(S_j)$ as $S_j,~1 \leq j \leq k$, form a partition of $T$ and the triangle inequality holds.
\end{proof}

The previous theorem shows that the approximation factor of even a naive algorithm as the one described by Algorithm \ref{alg:NaiveAlg} is bounded. However, the bound grows linearly with the number of agents $k$. This motivates us to look for more efficient algorithms to solve the task allocation problem at hand, in particular, algorithms that perform well as the number of agents become large. We provide two constant factor algorithms in the following section. 

\section{Constant Factor Approximation Algorithms for the Heterogeneous Task Allocation Problem}

\subsection{Cycle Splitting Approach}
Consider an instance of HTAP. In order to find an allocation of tasks to agents, we must allocate type-specific tasks among agents of the required type and allocate generic tasks among all agents. We approach the problem by handling these two allocations separately.

Given a set of tasks $T_i$ and a number $k$, let \textit{TaskSplitter} be any algorithm that splits the set of tasks $T_i$ into $k$ sub-tours $\{T_{i1}, T_{i2}, ..., T_{ik}\}$ within some factor $\beta$ of the optimal split (in the min-max sense). For example, Frederickson, Hecht and Kim \cite{frederickson1976approximation} give a tour splitting heuristic that takes a tour on a set of locations to be visited (first node of which is set as the start node), and a positive integer $k$ as input and gives a set of $k$ subtours (starting and ending at the start node) as the output. The cost of these subtours are within a factor of $1 + F - 1/k$ of the optimal min-max cost, where $F$ is the approximation factor of the algorithm used to generate the initial tour on all input locations. 

Consider the following algorithm to allocate a given set of tasks $T$ to a group of $k$ heterogeneous agents.

\begin{breakablealgorithm}
\caption{Min-Max Tour by \textsc{CycleSplit} Algorithm}
\label{alg:CycleSplitAlg}
\begin{algorithmic}[1]
\Procedure{CycleSplit}{$A,T,G$}
		\For{each agent type $i$, $1 \leq i \leq m$}
			\State \parbox[t]{\dimexpr 415pt}{Find Christofides' tour (starting and ending at $v_s$) on the set of type-specific tasks of type $i$.\strut}
			\State \parbox[t]{\dimexpr 415pt}{Use \textit{TaskSplitter} to get $m_i$ subtours (starting and ending at $v_s$) on set $T_i$.\strut}
			\State \parbox[t]{\dimexpr 415pt}{Allocate one subtour to each agent of type $i$. \strut}
		\EndFor
		\State \parbox[t]{\dimexpr 433pt}{Find Christofides' tour (starting and ending at $v_s$) on the set of generic tasks $T_0$.\strut}
		\State \parbox[t]{\dimexpr 433pt}{Use \textit{TaskSplitter} to split tour on $T_0$ into $k$ subtours (starting and ending at $v_s$), denoted by $\{R_1, R_2, \ldots, R_k\}$.\strut}
		\State \parbox[t]{\dimexpr 433pt}{Allocate subtour $R_j$ to agent $A_j, ~1\leq j \leq k$.\strut}
		\State \parbox[t]{\dimexpr 433pt}{Combine the type-specific task subtour and generic task subtour allocated to each agent.\strut}
		\State \parbox[t]{\dimexpr 433pt}{Return a tour for each agent.\strut}
\EndProcedure
\end{algorithmic}
\end{breakablealgorithm}

\begin{theorem}
\textsc{CycleSplit} is a $2\beta$-approximation algorithm for the Heterogeneous Task Allocation Problem (HTAP), where $\beta$ is the approximation factor of the algorithm \textit{TaskSplitter} used in steps 4 and 8.
\label{thm_2beta}
\end{theorem}

\begin{proof}
For $1 \leq j \leq k$, let $S_j^*$ be the allocation of tasks to agent $A_j$ under an optimal algorithm for HTAP. Then, $S_j^*$ can be expressed as $S_j^*=V_j^* \cup R_j^*$, where $V_j^*$ is the subset of type-specific tasks allocated to agent $A_j$ and $R_j^*$ is the subset of generic tasks assigned to agent $A_j$. Let $S_j = V_j \cup R_j$ be the allocation to agent $A_j$ by the \textsc{CycleSplit} algorithm, where $V_j$ is the subset of type-specific tasks allocated to agent $A_j$ and $R_j$ is the subset of generic tasks allocated to agent $A_j$. Let $C^*(\cdot)$ denote the cost to optimally tour a given set of tasks starting and ending from the node $v_s$. Let $C_{TS}(R_j)$ and $C_{TS}(V_j)$ denote the cost of the subtours $R_j$ and $V_j$ returned by the \textit{TaskSplitter} algorithm in steps 4 and 8 respectively. Let $C_{CS}(S_j)$ denote the cost of a tour on $S_j$ returned by the \textsc{CycleSplit} algorithm. Thus, the approximation factor of the \textsc{CycleSplit} algorithm is given by
\begin{align}
	R & = \frac{\max\limits_{1 \leq j \leq k} C_{CS}(S_j)}{\max\limits_{1 \leq j \leq k} C^*(S_j^*)} \leq \frac{\max\limits_{1 \leq j \leq k} \{C_{TS}(V_j) + C_{TS}(R_j)\}}{\max\limits_{1 \leq j \leq k} C^*(S_j^*)} \notag \\ 
	& \leq \frac{\max \limits_{1 \leq j \leq k} C_{TS}(V_j)}{\max \limits_{1 \leq j \leq k} C^*(S^*_j)} + \frac{\max \limits_{1 \leq j \leq k} C_{TS}(R_j)}{\max \limits_{1 \leq j \leq k} C^*(S^*_j)} \notag  \\ 
	& \leq \frac{\max\limits_{1 \leq j \leq k} C_{TS}(V_j)}{\max\limits_{1 \leq j \leq k} C^*(V_j^*)} + \frac{\max\limits_{1 \leq j \leq k} C_{TS}(R_j)}{\max\limits_{1 \leq j \leq k} C^*(R_j^*)},
\label{HTAP_R}
\end{align}
where we use the facts that $S_j = V_j \cup R_j$, $S^*_j = V^*_j \cup R^*_j$ and that the triangle inequality holds.

Consider a case with the same set of type-specific tasks as above, but with no generic tasks. Let $V'_1, V'_2, \ldots, V'_k$ denote the set of tasks allocated to agents under an optimal allocation (in the sense of minimizing the maximum cost to tour any subset). Since the \textsc{CycleSplit} algorithm allocates type-specific tasks independent of generic tasks, the allocation under this algorithm will be $V_1, V_2, \ldots, V_k$. Since the set of tasks with only type-specific tasks is a subset of the total set of tasks (including generic tasks), min-max cost for case with only type-specific tasks cannot exceed the min-max cost for the set of total tasks. Thus, 
\begin{equation}
\max\limits_{1 \leq j \leq k} C_{TS}(V_j) \leq \beta \max\limits_{1 \leq j \leq k} C^*(V'_j).
\label{eq:inter1}
\end{equation}

The inequality in the equation above follows from the fact that \textit{TaskSplitter} algorithm splits a tour into subtours that are within a factor $\beta$ of the optimal min-max cost. By the same argument as before, min-max cost to optimally tour $\{V'_j\},~1\leq j \leq k$ cannot exceed the min-max cost to optimally tour $\{V^*_j\},~1\leq j \leq k$, by the optimality of the partition $\{V'_j\},~1\leq j \leq k$. Thus, 
\begin{equation}
\max\limits_{1 \leq j \leq k} C^*(V'_j) \leq \max\limits_{1 \leq j \leq k} C^*(V^*_j).
\label{eq:inter2}
\end{equation}

Combining equations \eqref{eq:inter1} and \eqref{eq:inter2}, we get 
\begin{equation}
\max\limits_{1 \leq j \leq k} C_{TS}(V_j) \leq \max\limits_{1 \leq j \leq k} \beta C^*(V_j^*) \label{CS_1}.
\end{equation}

Using a similar procedure as before, this time comparing the allocation of a case with no type specific tasks against the set of all tasks, we get,
\begin{equation}
\max\limits_{1 \leq j \leq k} C_{TS}(R_j) \leq \max\limits_{1 \leq j \leq k} \beta C^*(R_j^*) \label{CS_2}.
\end{equation}

Substituting equations \eqref{CS_1} and \eqref{CS_2} in equation \eqref{HTAP_R}, we get
\begin{equation}
R \leq 2 \beta.
\end{equation}
\end{proof}

Frederickson, Hecht and Kim \cite{frederickson1976approximation} provide an algorithm \textit{SPLITOUR}\footnote{Frederickson, Hecht and Kim call the algorithm \textit{k-SPLITOUR} in their paper; however, we refer to it as \textit{SPLITOUR} in order to avoid confusing $k$ with the number of agents as we use it in this paper.} that splits a tour (starting and ending at a start-node) on a set of nodes into $k$ subtours each starting and ending at the start node. This algorithm has an approximation factor $1 + F - \frac{1}{k}$, where $F$ is the approximation factor of the algorithm used for constructing the initial tour.

\begin{corollary}
	\textsc{CycleSplit} algorithm using the \textit{SPLITOUR} algorithm from \cite{frederickson1976approximation} for task splitting is a $5-\frac{2}{k}$ approximation algorithm for the Heterogeneous Task Allocation Problem (HTAP).
\label{prop_5fac}
\end{corollary}

\begin{proof}
	The algorithm \textit{SPLITOUR} is a $\frac{5}{2}-\frac{1}{k}$ factor algorithm to split a tour using Christofides' algorithm into $k$ subtours (since Christofides provides an $F=\frac{3}{2}$ factor approximation for the initial tour). Equation \eqref{HTAP_R} can be written as
	\begin{equation}
	R \leq \Bigg(\frac{5}{2} - \frac{1}{\max\limits_{1 \leq i \leq m} m_i}\Bigg) + \Bigg(\frac{5}{2} - \frac{1}{k}\Bigg).
	\label{HTAP_GenEq}
	\end{equation}
	
	We can upper bound $\max\limits_{1 \leq i \leq m} m_i$ with $k$. Thus, 
	\begin{equation}
	R \leq 5 - \frac{2}{k}.
	\label{HTAP_eq}
	\end{equation}
\end{proof}

\begin{corollary}
	In the special instance of HTAP where there is exactly one agent of each type (i.e., the heterogeneous agents are distinct), \textsc{CycleSplit} is a $4-\frac{1}{k}$ approximation factor algorithm.
\label{cor_4fac}
\end{corollary}

\begin{proof}
	In this special case, $m_i = 1,~1\leq i \leq m$. Thus $\max\limits_{1 \leq i \leq m} m_i = 1$. So, equation \eqref{HTAP_GenEq} can be simplified as 
	\begin{equation*}
	R \leq \Bigg(\frac{5}{2} - 1 \Bigg) + \Bigg(\frac{5}{2} - \frac{1}{k}\Bigg) = 4 -\frac{1}{k}.
	\end{equation*}
\end{proof}

\begin{remark}
	From Corollary \ref{prop_5fac}, we see that regardless of the value of $k$, \textsc{CycleSplit} is a $5$-approximation algorithm for HTAP. Furthermore, from Corollary \ref{cor_4fac}, we can see that for the special instance of HTAP where all heterogeneous agents are distinct, \textsc{CycleSplit} is a $4$-approximation algorithm.
\end{remark}

Note that the \textsc{CycleSplit} algorithm while using \textit{SPLITOUR} as \textit{TaskSplitter} in lines 4 and 8, splits a TSP tour on the set of tasks into $k$ subtours evenly and allocates each subtour to one of the agents. While this approach works well when all agents have same the capacity (or balanced initial load), it may not perform as well when different agents have different capacity or initial load. In a case where agents have type-specific tasks that they must perform, each agent may have different task loads prior to allocation of generic tasks. In order to optimize the min-max tour cost, it is advantageous to take this into consideration. This motivates the formulation of an algorithm that takes into consideration the existing ``load'' of each agent from type-specific tasks when allocating generic tasks to agents. To address this, we propose a modified algorithm in the following section.

\subsection{Min-Max Splitting Approach for Heterogeneous Agents}
Based on the intuition gained from the \textsc{CycleSplit} algorithm, we propose a modified algorithm called \textsc{HeteroMinMaxSplit}. This algorithm considers the fact that agents may have different ``loads'' after the allocation of type-specific tasks. Specifically, instead of splitting a tour on the set of generic tasks into nearly equal segments of a given size, the generic tasks are allocated to agents while taking into consideration the cost incurred by the agent to tour its specific tasks. In the \textsc{HeteroMinMaxSplit} algorithm, we allocate tasks to agents in three phases.

\begin{itemize}
	\item \textbf{Phase 1}: Type-specific task allocation
	\item \textbf{Phase 2}: Generic task allocation (accounting for Phase 1 allocation)
	\item \textbf{Phase 3}: Rebalancing tasks within agents of each type.
\end{itemize} 

In Phase 1, type-specific tasks are allocated among agents of the associated type as in \textsc{CycleSplit}. The generic tasks are allocated among all agents in Phase 2. Unlike in \textsc{CycleSplit}, however, the allocation during Phase 2 tries to balance the total cost incurred by agents by taking into account the existing allocation to agents after Phase 1. Phase 3 is a post-processing step aimed to further balance the load among agents. After Phase 2, all tasks allocated to agents of the same type can be done by all agents in that type. Thus, we pool all tasks allocated to agents of the type $i$ into a set $T'_i$ and re-allocate them to agents of type $i$ by splitting the Christofides' tour on $T'_i$ into $m_i$ sub-tours (Algorithm \textit{SPLITOUR} may be used for this purpose \cite{frederickson1976approximation}). Balancing the load in this fashion can lower the min-max tour cost by distributing the load more evenly among agents, as we will show later.

We present the algorithm in two steps. We first give an algorithm \textsc{HeteroSplit} that takes an additional input $\lambda$, which we take to be a rational number. This algorithm finds tours (if they exist) for each agent such that cost for each agent to complete its tour is less than $\lambda$. The \textsc{HeteroMinMaxSplit} algorithm then performs a binary search on $\lambda$, running \textsc{HeteroSplit} in each iteration, to find the best set of tours for the agents.

\begin{breakablealgorithm}
	\caption{Heterogeneous Task Split within bound $\lambda$}
	\label{alg:HeteroSplit}
	\begin{algorithmic}[1]
		\LeftComment Let $\lambda$ be a rational parameter that denotes the desired upper bound for the tour length for any agent
		\Procedure{HeteroSplit}{$A,T,G, \lambda$}
		\LeftComment \textbf{Phase 1:} Type-specific task allocation
		\For{each agent type $i$, $1 \leq i \leq m$}
			\State \parbox[t]{\dimexpr 415pt}{Find Christofides' tour (starting and ending at $v_s$) on the set of type-specific tasks of type $i$.\strut}
			\State \parbox[t]{\dimexpr 415pt}{Use \textit{TaskSplitter} to split the tour on $T_i$ into $m_i$ subtours (starting and ending at $v_s$).\strut}
			\State \parbox[t]{\dimexpr 415pt}{Allocate one subtour to each agent of type $i$. \strut}
		\EndFor
		\LeftComment \textbf{Phase 2:} Generic task allocation
		\State \parbox[t]{\dimexpr 433pt}{Mark all agents as free. \strut}
		\State \parbox[t]{\dimexpr 433pt}{Remove from $G$, all vertices in $T_i,~1\leq i \leq m$ and all edges incident on these vertices. Denote the resulting graph by $G'$.\strut}
		\State \parbox[t]{\dimexpr 433pt}{Mark all tasks in graph $G'$ as unallocated. Find Christofides' tour starting and ending at $v_s$ on nodes in $G'$. Denote the tour by $H$.\strut}
		\State \parbox[t]{\dimexpr 433pt}{Consider the next unallocated task, say $t$, along $H$ starting from $v_s$. For each free agent $j$, find cost to tour (starting and ending at $v_s$) all tasks allocated to it along with $t$ using Christofides' algorithm. Select agent with the minimum cost provided the minimum cost is less than $\lambda$. If no free agent can add the task to its tour without exceeding cost $\lambda$, return failure. \strut}
		\State \parbox[t]{\dimexpr 433pt}{Allocate $t$ to the selected free agent. Keep allocating unallocated tasks along $H$ to the selected agent as long as adding the task does not cause the agent's tour cost to exceed $\lambda$.\strut}
		\State \parbox[t]{\dimexpr 433pt}{Remove the tasks allocated in the previous step from the set of unallocated tasks and mark agent as busy.\strut}
		\State \parbox[t]{\dimexpr 433pt}{Go to step 10 if the set of unallocated tasks is non-empty and free agents remain. \strut}
		\LeftComment \textbf{Phase 3:} Rebalancing within agent types
		\For{each agent type $i$, $1 \leq i \leq m$}
			\State \parbox[t]{\dimexpr 415pt}{Deallocate tasks from all $m_i$ agents of type $i$. Let $T'_i$ denote the set of tasks deallocated from agents of type $i$. \strut}
			\State \parbox[t]{\dimexpr 415pt}{Find Christofides' tour on $T'_i$ to get tour $H'$.\strut}
			\State \parbox[t]{\dimexpr 415pt}{Run \textit{TaskSplitter} on $H'$ to split it into $m_i$ subtours.\strut}
			\State \parbox[t]{\dimexpr 415pt}{Allocate one subtour to each agent of type $i$.\strut}
		\EndFor
		\State \parbox[t]{\dimexpr 433pt}{Return tours for each agent. \strut}
		\EndProcedure
	\end{algorithmic}
\end{breakablealgorithm}

A binary search can be performed to find the optimal value for $\lambda$ to be provided as a parameter to \textsc{HeteroSplit}. The value cannot be less than than twice the largest edge from the start node to any task location, i.e., $2 \max\limits_{t \in T}d(v_s,t)$, where $d(v_s, t)$ is the distance from the start node $v_s$ to the location of a task $t$. Also, $\lambda$ cannot exceed $\max\limits_{1 \leq i \leq k} C(T_i) + C(T_0)$, where $C(\cdot)$ is the cost to do a Christofides' tour on the given set of tasks. Hence, we can perform a binary search within this window to find the optimal value of $\lambda$.

\begin{breakablealgorithm}
	\caption{Min-Max Tour by \textsc{HeteroMinMaxSplit} Algorithm}
	\label{alg:HeteroMinMaxSplit}
	\begin{algorithmic}[1]
		\Procedure{HeteroMinMaxSplit}{$A,T,G$}
		\State \parbox[t]{\dimexpr 433pt}{Do binary search in the interval $[2 \max\limits_{t \in T} d(v_s,t),~\max\limits_{1 \leq i \leq k} C(T_i) + C(T_0)]$ to find smallest value of $\lambda$  for which \textsc{HeteroSplit}($A,T,G,\lambda$) returns a set of valid tours. \strut}
		\State \parbox[t]{\dimexpr 443pt}{Return the set of tours returned by \textsc{HeteroSplit}($A,T,G,\lambda$). \strut}
		\EndProcedure
	\end{algorithmic}
\end{breakablealgorithm}

\subsection{Bound on Approximation Ratio of \textsc{HeteroMinMaxSplit} Algorithm}
For the purpose of applying \textit{SPLITOUR} from Frederickson, Kim and Hecht's work, we now provide a brief sketch of the algorithm \textit{SPLITOUR} \cite{frederickson1976approximation}. It takes as input a set of $n$ nodes labeled $v_1$ through $v_n$ (with the initial node as the start node) and a positive number $k$. First, the algorithm constructs a Christofides' tour on all the nodes. Let $L$ be the cost of this tour and let $c_{max}$ be the maximum direct distance of any node from the start node, i.e., $c_{max} = \max \limits_{1 < i \leq n} d(v_1, v_n)$. For $1 \leq j <k$, it finds the largest vertex $v_{p(j)}$ along the tour such that the cost to traverse the tour from the start node to $v_{p(j)}$ does not exceed $\frac{j}{k} (L-2c_{max})+c_{max}$. It returns $k$ subtours $(v_1,\ldots,v_{p(1)},v_1)$, $(v_1,v_{p(1)+1},\ldots, v_1)$, $\ldots$, $(v_1,v_{p(k-1)+1},\ldots,v_n,v_1)$, each of which has cost less than $\frac{1}{k} (L-2c_{max})+2c_{max}$. Frederickson, Hecht and Kim \cite{frederickson1976approximation} show that if the length of the subtours do not exceed $\frac{1}{k} (L-2c_{max})+2c_{max}$, then the min-max cost of the subtours is no worse than a factor $(\frac{5}{2} - \frac{1}{k})$ of the optimal min-max cost. 

The following proposition leverages the splitting heuristic summarized above to establish results on the performance of our algorithms.

\begin{theorem}
	\textsc{HeteroMinMaxSplit} is a $(5-\frac{2}{k})$-approximation algorithm for the Heterogeneous Task Allocation Problem (HTAP).
	\label{thm:HMMS}
\end{theorem}

\begin{proof}
	Let $\lambda_1$ be the maximum cost among all the subtours returned by running \textit{SPLITOUR} on Christofides' tours on type-specific tasks after line 6 of the \textsc{CycleSplit} algorithm, i.e., $\lambda_1 = \max \limits_{1 \leq j \leq k} C_{TS}(V_j)$, where $V_j$ denotes the set of type-specific tasks allocated to agent $A_j$ and $C_{TS}(\cdot)$ denotes the cost associated with the subtour returned by \textit{SPLITOUR}. Note that $\lambda_1$ is also the maximum cost after Phase 1 of the \textsc{HeteroSplit} algorithm.
	
	Running \textit{SPLITOUR} on a Christofides' tour on the set of generic tasks $T_0$ to split it into $k$ subtours will returns subtours of length (or cost) no more than $\frac{1}{k}(L-2c_{max}) + 2c_{max}$, where $L$ is the length (or cost) of the initial Christofides' tour on $T_0$, and $c_{max}$ is the maximum direct distance of any node in $T_0$ from the start node. No matter how each of these subtours are assigned to agents (such that each agent gets one subtour), the total cost is no more than $\lambda_1 + \frac{1}{k}(L-2c_{max}) + 2c_{max}$. 
	
	Set $\lambda$ to be $\lambda_1 + \frac{1}{k}(L-2c_{max}) + 2c_{max}$. Consider \textsc{HeteroSplit} with this value of $\lambda$. In Phase 1, \textsc{HeteroSplit} allocates type-specific tasks to agents, same as steps 2-6 of \textsc{CycleSplit}. The maximum cost of any agent's tour after Phase 1 is $\lambda_1$. In Phase 2, the \textsc{HeteroSplit} algorithm computes a Christofides' tour on the set of generic tasks. The algorithm selects an agent that has minimum cost to complete its current allocated tasks in addition to the next task along the Christofides' tour on generic tasks. The algorithm then allocates tasks along the tour to the selected agent as long as the cost does not exceed $\lambda$. Regardless of which agent gets selected first, the set of tasks that are allocated to the selected agent in \textsc{HeteroSplit} will contain the tasks in the first subtour generated by \textit{SPLITOUR}. The allocation by \textsc{HeteroSplit} to the first selected agent may contain more tasks than the first subtour of \textit{SPLITOUR}, but not less. Thus, after allocating tasks to the first agent, the number of tasks left to be allocated to the remaining $k-1$ agents in \textsc{HeteroSplit} is no more than the number of tasks left in the $k-1$ subtours to be allocated to the remaining $k-1$ agents in \textit{SPLITOUR}.
	
	In \textit{SPLITOUR}, after the allocation of the first subtour to the an agent, it is guaranteed that the remaining $k-1$ subtours can be allocated to the remaining $k-1$ agents, since the allocation is possible regardless of the order in which agents are selected. Given that the set of tasks left for allocation in \textsc{HeteroSplit} is no bigger than the set of tasks left to be allocated in \textit{SPLITOUR}, it is guaranteed that the remaining tasks can be allocated to other $k-1$ agents in \textsc{HeteroSplit} algorithm for this value of $\lambda$. Thus, by inducting on the number of agents to which tasks have been allocated along the initial tour, the \textsc{HeteroSplit} algorithm is guaranteed to return a feasible solution to HTAP for $\lambda =  \lambda_1 + \frac{1}{k}(L-2c_{max}) + 2c_{max}$. 
	
	Since \textsc{HeteroSplit} checks tour cost against $\lambda$ before further allocation of tasks to an agent, it guarantees that the tour costs for all agents is no more than $\lambda$. Thus, the approximation factor is given by, 
	
	\begin{align}
	R & \leq \frac{\lambda}{\max\limits_{1 \leq j \leq k} C^*(S^*_j)} = \frac{\lambda_1}{\max\limits_{1 \leq j \leq k} C^*(S^*_j)}+ \frac{\frac{1}{k}(L-2c_{max}) + 2c_{max}}{\max\limits_{1 \leq j \leq k} C^*(S^*_j)} \notag \\
	& \leq \frac{\max \limits_{1 \leq j \leq k} C_{TS}(V_j)}{\max\limits_{1 \leq j \leq k} C^*(V^*_j)}+ \frac{\frac{1}{k}(L-2c_{max}) + 2c_{max}}{\max\limits_{1 \leq j \leq k} C^*(R^*_j)} \label{subset}\\
	& \leq \Bigg(\frac{5}{2} - \frac{1}{\max\limits_{1 \leq i \leq m} m_i}\Bigg) + \Bigg(\frac{5}{2} - \frac{1}{k}\Bigg) \label{eq:Fred}\\
	& \leq 5 - \frac{2}{k},
	\end{align}
	where equation \eqref{subset} is given by the triangle inequality and the fact that $V^*_j$ and $R^*_j$ are subsets of $S^*_j$. Equation \eqref{eq:Fred} is obtained from the results of Frederickson, Hecht and Kim \cite{frederickson1976approximation}, summarized prior to this theorem.
	
	The \textsc{HeteroMinMaxSplit} algorithm performs a binary search over $\lambda$ and is thus guaranteed to find the $\lambda$ used above by \textsc{HeteroSplit}. Thus, \textsc{HeteroMinMaxSplit} is a $5-\frac{2}{k}$ approximation algorithm for HTAP. 
\end{proof}

\begin{corollary}
	\textsc{HeteroMinMaxSplit} is a $(4-\frac{1}{k})$-approximation algorithm for the special instance of HTAP where all agents are distinct.
	\label{cor_HMMS}
\end{corollary}

\begin{proof}
	In this special instance, $m_i = 1$, for $1\leq i \leq m$. Substituting in equation \eqref{eq:Fred}, we get the approximation factor $R \leq 4 - \frac{1}{k}$. 
\end{proof}

From Theorem \ref{thm:HMMS}, we see that \textsc{HeteroMinMaxSplit} is a $5$-approximation algorithm regardless of the value of $k$. In this special instance where all heterogeneous agents are distinct, we see from Corollary \ref{cor_HMMS} that \textsc{HeteroMinMaxSplit} is a $4$-approximation algorithm regardless of the value of $k$.

\subsection{Examples Illustrating \textsc{HeteroMinMaxSplit} and \textsc{CycleSplit} Algorithms}
\textsc{HeteroMinMaxSplit} takes into consideration the ``load'' from type-specific tasks while allocating generic tasks and can outperform \textsc{CycleSplit} when the type-specific tasks are not uniformly distributed among agents. We illustrate this through the following examples.

\begin{example}
	Consider the scenario where three agents are collectively required to do nine tasks. Agent $A_1$ is a type 1 agent, $A_2$ is a type 2 agent and $A_3$ is a type 3 agent. All agents are located at node $v_s$. The start node $v_s$ is located at distance $d' > 2$ from nodes $V_1$ and $V_2$ as shown in Figure \ref{fig:ex1}. Nodes $A$, $B$ and $C$ are located at unit distance from node $V_1$, and node $V_2$ is located at a distance $d > 3$ (such that the triangle inequality is satisfied) from node $C$ as shown in Figure \ref{fig:ex1}. Distances between nodes not directly connected are defined as the sum of distances along the shortest path traversed to reach the node (the graph can be considered as a road network that the agents can traverse).
	
	\begin{figure}[htb]
		\begin{center}
			\begin{tikzpicture}
			[scale=.13, inner sep=4pt, minimum size=7pt, auto=center]
			\node [circle, draw, fill=gray!40, label=180:$A$, label=270:$t_1~t_2$](a) at (-26, 6)  {};
			\node [circle, draw, fill=gray!40, label=5:$C$](c) at (-10, 6)  {};
			\node[label=above:{$t_5~t_6$}] at (-6,0) {};
			\node [circle, draw, fill=gray!40, label=0:$B$, label=270:$t_3~t_4$](e) at (-18, 0)  {};
			\node [circle, draw, fill=gray!90, label=90:$V_1$](b) at (-18, 6)  {};
			\node[label=above:{$t_7~t_8$}] at (-18,10) {};
			\node [circle, draw, fill=gray!90, label=90:$V_2$, label=270:$t_9$](d) at (26, 6)  {};
			\node [circle, draw, fill=black, label=90:$v_s$](f) at (4, 18)  {};
			\draw (a) -- (b) node[midway, above]{$1$};
			\draw (b) -- (c) node[midway, below]{$1$};
			\draw (b) -- (e) node[midway, left]{$1$};
			\draw (c) -- (d) node[midway, below]{$d$};
			\draw (f) -- (b) node[midway, above]{$d'$};
			\draw (f) -- (d) node[midway, above]{$d'$};
			\end{tikzpicture}
		\end{center}
		\caption{Task locations in Example \ref{ex1}}
		\label{fig:ex1}
	\end{figure}
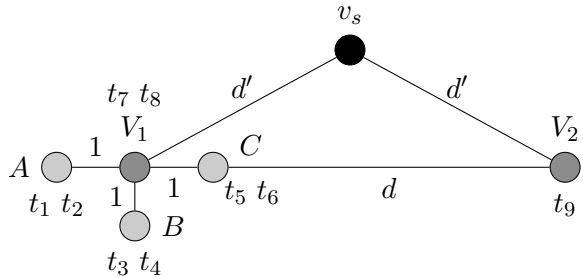
	
	Tasks in set $T_0 = \{t_1,t_2,t_3,t_4,t_5,t_6\}$ are generic tasks. Tasks $\{t_1, t_2\}$ are located at node $A$, $\{t_3, t_4\}$ are located at node $B$, and $\{t_5, t_6\}$ are located at node $C$. Agent $A_1$ has one type-specific task $t_7$, i.e. $T_1 = \{t_7\}$. Agent $A_2$ has one type-specific task $t_8$, i.e. $T_2 = \{t_8\}$. Tasks $t_7$ and $t_8$ are located at node $V_1$. Agent $A_3$ has one type-specific task $t_9$, i.e., $T_3 = \{t_9\}$ and $t_9$ is located at node $V_2$.  
	
	Both \textsc{CycleSplit} and \textsc{HeteroMinMaxSplit} algorithms first allocate the type-specific tasks to the respective agents. Thus, agent $A_1$ is allocated task $t_7$, agent $A_2$ is allocated task $t_8$ and agent $A_3$ is allocated task $t_9$.
	
	\textsc{CycleSplit} algorithm splits the generic tasks at nodes $A$, $B$ and $C$ to the three agents. Thus, agent $A_1$ gets allocated tasks $\{t_1,t_2\}$, agent $A_2$ gets allocated tasks $\{t_3,t_4\}$, and agent $A_3$ gets allocated tasks $\{t_5, t_6\}$. The min-max cost is thus $\max \{2d'+2, 2d'+2, 2d'+d+1\} = 2d'+d+1$.
	
	The \textsc{HeteroMinMaxSplit} algorithm allocates tasks in Phase 2 based on the allocation done in Phase 1. Thus, for $\lambda = 2d'+ 4 < 2d'+d+1$, \textsc{HeteroSplit} allocates all tasks at the nodes $A$, $B$ and $C$ to agents $A_1$ and $A_2$ (which already travel to node $V_1$ to do their type-specific tasks). Thus, tasks $\{t_1, t_2\} $ located at node $A$ are allocated to agent $A_1$, and tasks $\{t_3, t_4, t_5, t_6\} $ located at nodes $B$ and $C$ are allocated to agent $A_2$. Since only one agent of each type is present, agents retain their allocation after Phase 3. The min-max cost is thus $\max \{2d'+2, 2d'+4, 2d'\} = 2d'+4$.
	
	For large values of $d'$ and $d$, \textsc{HeteroMinMaxSplit} gives a factor of 2 over \textsc{CycleSplit} on the min-max tour cost.
	\label{ex1}
\end{example}

Note that in the above example, Phase 3 does not improve the min-max cost. We illustrate the benefit of Phase 3 (rebalancing ``load'' between agents of a given type) of the \textsc{HeteroMinMaxSplit} algorithm in the following example. 

\begin{example}
	Consider a scenario with two type 1 agents $A_1$ and $A_2$ located at the start node $v_s$. Let type-specific tasks for type 1 agents $T_1=\{t_1,t_2\}$ be located at node $A$, which is at a distance of $1$ unit from $v_s$ as shown in Figure \ref{fig:phase3}. Generic tasks $T_0 = \{t_3,t_4\}$ are located at node $B$ which is at a distance of $1$ from $v_s$ as shown in Figure \ref{fig:phase3}.
	
	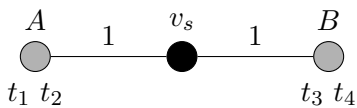
\begin{figure}[htb]
		\begin{center}
			\begin{tikzpicture}
			[scale=.13, inner sep=4pt, minimum size=7pt, auto=center]
			\node [circle, draw, fill=gray!60, label=90:$A$, label=270:$t_1~t_2$](a) at (-15, 6)  {};
			\node [circle, draw, fill=gray!60, label=90:$B$, label=270:$t_3~t_4$](b) at (15, 6)  {};
			\node [circle, draw, fill=black, label=90:$v_s$](d) at (0, 6)  {};
			
			\draw (a) -- (d) node[midway, above]{1} ;
			\draw (d) -- (b) node[midway, above]{1} ;
			
			\end{tikzpicture}
		\end{center}
		\caption{Task locations in Example \ref{ex2}}
		\label{fig:phase3}
	\end{figure}
	In Phase 1, one type-specific task gets allocated to each of the agents, say $t_1$ gets allocated to agent $A_1$ and $t_2$ gets allocated to agent $A_2$. In Phase 2, one generic task gets allocated to each of the agents, say $t_3$ gets allocated to agent $A_1$ and $t_4$ gets allocated to agent $A_2$. So at the end of Phase 2, both agents need to visit nodes $A$ and $B$ to complete their allocated tasks. 
	
	In Phase 3, all tasks allocated in the previous phases to agents of the same type are deallocated and redistributed amongst them. In this phase, tasks $\{t_1,t_2,t_3,t_4\}$ are pooled and the \textit{TaskSplitter} algorithm \textit{SPLITOUR} is run on this set of tasks. Thus, one agent, say $A_1$, gets tasks at node $A$ and the other agents $A_2$ gets tasks at node $B$. In this example, Phase 3 reduces the tour cost from $4$ to $2$ for both the agents, thus reducing the min-max cost by a factor of 2.
	\label{ex2}
\end{example}

\begin{example}
	Consider two agents: $A_1$ of type 1 and $A_2$ of type 2. Task $t_1$ is of type 1 and must be performed by agent $A_1$, and tasks $t_2,~t_3$ are generic tasks that can be done by either of the agents, i.e., $T_1=\{t_1\},~T_2 =\emptyset$, $T_0=\{t_2,t_3\}$. Task $t_1$ is located at node $A$. Tasks $t_2$ and $t_3$ are located at node $B$. Agents start from start node $v_s$. 
	
	\begin{figure}[htb]
		\begin{center}
			\begin{tikzpicture}
			[scale=.13, inner sep=4pt, minimum size=7pt, auto=center]
			\node [circle, draw, fill=gray!60, label=90:$A$, label=270:$t_1$](a) at (-15, 6)  {};
			\node [circle, draw, fill=gray!60, label=90:$B$, label=270:$t_2~t_3$](b) at (15, 6)  {};
			\node [circle, draw, fill=black, label=90:$v_s$](d) at (0, 6)  {};
			
			\draw (a) -- (d) node[midway, above]{1} ;
			\draw (d) -- (b) node[midway, above]{1} ;
			
			\end{tikzpicture}
		\end{center}
		\caption{Task locations in Example \ref{ex3}}
		\label{fig:ex2}
	\end{figure}
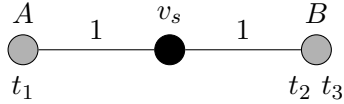
	
	\textsc{CycleSplit} algorithm allocates type-specific task $t_1$ to agent $A_1$ and splits generic tasks $\{t_2,t_3\}$ among agents $A_1$ and $A_2$. Hence, $A_1$ is allocated tasks $\{t_1,t_2\}$ and $A_2$ is allocated $\{t_3\}$. The tour costs for agents $A_1$ and $A_2$ are 4 and 2 respectively, and thus, the min-max tour cost is $\max\{4,2\} = 4$.
	
	\textsc{HeteroMinMaxSplit} algorithm allocates type-specific task $t_1$ to agent $A_1$ in Phase 1. In Phase 2, generic tasks $\{t_2,t_3\}$ are split among agents $A_1$ and $A_2$ based on remaining capacity. For $\lambda = 2$, \textsc{HeteroSplit} allocates $\{t_1\}$ to agent $A_1$ and $\{t_2,t_3\}$ to agent $A_2$. The min-max tour cost for the tour returned by \textsc{HeteroMinMaxSplit} is thus $\max\{2,2\} = 2$.  
	\label{ex3}
\end{example}

In this example, \textsc{HeteroMinMaxSplit} algorithm gives an improvement by a factor of 2 over \textsc{CycleSplit}. The following example extends it to a general case with $k$ agents.

\begin{example}
	Consider the following generalization of the previous example for $k$ agents. Agents $\{A_1, A_2, \ldots, A_k\}$ are located at start node $v_s$. Tasks $T_i=\{t_i\}$ located at node $V_i$ has to performed by $A_i$ for $i = 1, 2, \ldots, k-1$, and tasks $T_0=\{t_k, t_{k+1}, \ldots, t_{2k-1}\}$ located at node $V_k$ are generic tasks and can be completed by any agent. Let $d(\cdot,\cdot)$ denote the distance function. The distance between the start node $v_s$ and any other node is $1$ unit, i.e., $d(v_s,V_i) = 1$, for $1 \leq i \leq k$. The distance between two nodes $V_i$ and $V_j$ is given by $d(V_i,V_j)=d(V_i,v_s)+d(v_s,V_j) = 2$ if $i \neq j$ and zero if $i=j$. In this case, \textsc{CycleSplit} algorithm allocates tasks $\{t_i,t_{k+i-1}\}$ for agents $A_i,~1\leq i \leq k-1$ and task $\{t_{2k-1}\}$ to agent $A_k$. This allocation has a min-max tour cost of $4$. \textsc{HeteroMinMaxSplit} allocates task $\{t_i\}$ to agent $A_i$ for $1 \leq i \leq k-1$ and tasks $\{t_k, t_{k+1}, \ldots, t_{2k-1}\}$ to agent $A_k$. This allocation brings down the min-max tour cost to $2$.
\end{example}

\section{Summary}
In this work, we considered the Heterogeneous Task Allocation Problem (HTAP) where we aim to allocate tasks to heterogeneous agents subject to agent-task compatibility while minimizing the min-max tour cost. We provide two approximation algorithms to solve the Heterogeneous Task Allocation Problem (HTAP). We first propose a $2\beta$ approximation algorithm \textsc{CycleSplit}, where $\beta$ is the approximation factor of the algorithm used by \textsc{CycleSplit} to split tours. We then use the \textsc{CycleSplit} algorithm to develop the \textsc{HeteroMinMaxSplit} algorithm which has an approximation ratio $5 - \frac{2}{k}$, where $k$ is the number of agents available. The \textsc{HeteroMinMaxSplit} algorithm works in three phases and allocates tasks to agents in a ``balanced'' way. An interesting extension to the work would be to study the case of allocation when tasks can be done by different subsets of agents (as opposed to agents of only a certain type).

\bibliographystyle{IEEEtran}
\bibliography{refs}
\end{document}